\documentclass[11pt,a4paper]{article}

\usepackage{graphicx,setspace}
\usepackage{psfrag}
\usepackage{amsfonts}
\usepackage{amsfonts,amsmath, amssymb, amsthm}
\usepackage{mathrsfs}
\usepackage{epstopdf}
\usepackage{float}
\usepackage{lineno}
\usepackage{color}
\usepackage{subfigure}
\usepackage{cases}
\usepackage[T1]{fontenc}
\usepackage[english]{babel}
\usepackage{algorithm,algorithmicx}
\usepackage{algpseudocode}
\usepackage{enumerate}
\usepackage{subeqnarray}
\newtheorem{lemma}{Lemma}[section]
\newtheorem{thm}{\bf Theorem}[section]

\title{A New Hybrid Consensus Protocol: Deterministic Proof Of Work}

\author{\uppercase{Zhuan Cheng}$^a$,
\uppercase{Gang Wu$^a$, Hao Wu$^a$, Muxing Zhao$^b$,\\ Liang Zhao$^a$, and Qingfeng Cai$^a$}}

\begin{document}
\maketitle

\medskip

\begin{abstract}
The Decentralized-Consistent-Scale (DCS) Triangle defines three dimensions that illustrate the tradeoffs of the blockchain consensus mechanism. 
In this paper, we propose a new hybrid consensus protocol, called Deterministic Proof of Work (DPoW), which can reach high levels of scalability and consistency without significant reduction to decentralization. Our protocol introduces a Map-reduce PoW mining algorithm to perform alongside Practical Byzantine Fault Tolerance (PBFT) verification, which together allow for transactions to be confirmed immediately, largely improving scalability. In addition, the protocol is designed such that forking cannot occur, ensuring strong consistency and security against a multitude of attacks. The Map-reduce PoW mining process ensures that no single entity can control the network, guaranteeing decentralization. We analyzed the security of our protocol by evaluating the possibility of double spending attacks, and furthermore, conducted experiments which demonstrate our claims.
\end{abstract}

\baselineskip=15pt


\section{Introduction}

For years, blockchain systems have struggled to optimize the Decentralized-Consistent-Scale (DCS) Triangle \cite{DCS} of the consensus mechanism. Specifically, out of decentralization, consistency, and scalability, only two can be achieved while sacrificing the third . Proof of Work (PoW) is used to achieve high levels of decentralization and consistency \cite{POW} while suffering severe scalability issues. Delegated Proof of Stake (DPoS) \cite{DPOS} and Practical Byzantine Fault Tolerance (PBFT) \cite{PBFT2}, on the other hand, are used to achieve high levels of consistency and scalability at the cost of the reduction of decentralization.

The concept of the PoW protocol was first detailed by Cynthia Dwork and Moni Naor \cite{HC1} in 1993, and in 1999, Markus Jakobsson \cite{POW} described the utilization of this protocol to produce a piece of data, termed the ``Proof of Work," that is controllably difficult to create but simple to verify. Producing a proof of work can be a random process with low probability of success such that extensive trial and error is required before a valid proof of work is generated. Proof of work is used in Bitcoin for block generation through the process of "mining", and the security of the network depends upon it \cite{BTC2}.

Currently, Bitcoin \cite{BTC1} operates at a rate of roughly $3$ to $7$ transactions per second, while Ethereum operates at approximately $7$ to $15$ per second. These rates are sub-optimal for the size of their growing userbases. This slow transaction rate bottleneck occurs due to three factors. Firstly, transactions require an extended period of time before they can be confirmed. In a PoW based blockchain, the confirmation of a transaction can be quantized by the number of blocks that follow the block it is contained in. Additionally, attacks such as the Finney attack, race attack, and 51\% attack mean transactions with little or no confirmations are vulnerable to reversal with little cost to the attacker. Due to this, the system is designed such that the more confirmations a transaction gets, the harder it becomes to reverse, while also increasing confirmation time, especially for transactions of large amounts. Block generation time must also be controlled to be long enough such that the process cannot occur too quickly and risk loss of consistency, thus further lowering transaction rates. Decreasing the block generation time will increase the chance of forks, thereby increasing the risk of alternative history attacks, which may result in double spending. Finally, the block spread rate is linearly related to the size of the block, thus incentivizing miners to generate small blocks in order to gain the advantage in the proof of work competition. 

The PoW protocol is also criticized for causing excessive waste of energy, which is primarily a result of the needless repeated calculations performed by every node. Specifically, each node must do the exact same calculation process as every other node to solve for the PoW of each block \cite{SM1}. In our protocol, we avoid this unnecessary repetition by dividing the work among nodes, thus reducing overall calculation times and power consumption.

This paper is principally concerned with the improvement of the blockchain consensus mechanism, which is vital for improving the overall performance of a blockchain. Here we propose a new protocol, called ``Deterministic Proof of Work (DPoW), which is characterized by separating the consensus process into a puzzle solving step and a verification step. This two step process combines the advantages of both PoW and PBFT, and works as follows:

Firstly, during each block generation round, miners perform the puzzle solving step to generate new blocks. We propose a Map-reduce PoW algorithm that optimizes energy consumption through a process called ``sharding,'' which essentially divides PoW calculations between nodes such that nodes do not need to individually process the entirety of the searching space. By separating the searching space into multiple ``shards'' and distributing the work to miners, each node only finds a proof of work for their small portion of the searching space, ensuring that nodes have very little chance of unnecessary repeated work. As a result, the time to calculate the hash of a block is decreased, which considerably reduces the total computations performed by the entire network. Therefore, through this method, we can increase the scalability of a blockchain while maintaining decentralization and consistency. Our experiments show that integrating Map-reduce PoW directly results in a significant improvement in the PoW solving time per block.

In the following verification step, newly generated blocks are sent to an elected verifier group during each block generating time. The verifier group then performs the PBFT consensus protocol on the blocks, which will in turn decide upon a single valid block to be appended to the chain. Lamport states that in distributed fault-tolerance systems, the maximum tolerance threshold of faulty members is $1/3$ \cite{PBFT4} of the total members. Experimentally, we confirmed that our protocol does indeed support this conclusion.

Due to the deterministic finality of our Deterministic Proof of Work protocol, we can ensure that upon a block's validation, it immediately becomes irreversible in the blockchain history. This means that there is no limit to the speed of block production, unlike the original PoW protocol. This is possible because our system is designed such that forks cannot develop, and thus there only ever exists a single chain in the entire network. By coupling both consensus mechanisms, a blockchain can have increased scalability as consensus no longer depends on a race of computing power and transactions can be validated much faster. Transaction validation times can be reduced to mere seconds compared to the hours required for bitcoin to accomplish the same task. This also retains a high level of consistency through the implementation of a verifier group via PBFT. Finally, it can perform these steps without significant reduction to decentralization as the block generation process remains uncontrolled by any central authority.

In summary, we make the following contributions:
\begin{enumerate}
\item  We construct a new consensus algorithm which is a hybrid of the PoW and PBFT consensus mechanisms.
\item  We propose a map-reduce mining protocol.
\item  We modify the PBFT algorithm to fit in our consensus protocol.
\item  We prove the security of our hybrid protocol against double spending attacks.
\item  We provide experimental data which measures the efficiency and security of our consensus algorithm.
\end{enumerate}

\section{Bitcoin Blockchain Consensus Protocol: Proof of Work}

The underlying database structure for transactions in Bitcoin and many other digital cryptocurrencies is a decentralized ledger. This ledger is maintained by anonymous parties, called ``miners,'' who execute a consensus protocol that maintains a data structure called the blockchain. A blockchain is a linked list of blocks, with each block containing the hash of the previous block in the chain. Transactions are bundled into each of these blocks. All full nodes in the network have their own copy of the blockchain in order to maintain the security of the system. The Bitcoin protocol requires that a valid ``proof of work''  must be provided for each new generated block for it to be appended to the chain. A proof of work is the solution to a cryptographic puzzle which contains the previous block hash, the Merkle root \cite{MR1} of the valid transactions in the current block, and a special transaction called a coinbase which rewards the miner for solving the cryptographic puzzle. A valid proof of work satisfies the following formula:
\begin{equation}
\begin{aligned}
{\rm{H}}\left( {\rm{PreHash}},{\rm{Diff}},{\rm{Time}},{\rm{TxRoot}},{\rm{Nonce}} \right) \le Target.
\end{aligned}
\end{equation}

In this expression, $\rm{H}$ is a cryptographic hash function (SHA-256 in bitcoin, SHA-3 in our paper), ${\rm{Target}}$ is the max possible hash value which makes the proof of work valid, ${\rm{PreHash}}$ is the hash of the previous block's header, ${\rm{Diff}}$ is the difficulty value which scales the ${\rm{Target}}$ within the hash value range, $\rm{Time}$ is the UNIX timestamp, ${\rm{TxRoot}}$ is the Merkle root \cite{MR1} of all transactions contained in the block, and $\rm{Nonce}$ is an integer that can be used to brute-force the hash inequality. The miner mines the block by incrementing the nonce until a valid proof of work has been discovered. In Fig. \ref{puzzle}, we show how the proof of work is organized. Once the block is mined successfully, the block cannot be changed without redoing the work due to the second pre-image resistance of the cryptographic hash function \cite{KELSEY}.

\begin{figure}[htbp]
  \centering
  \includegraphics[width=5in]{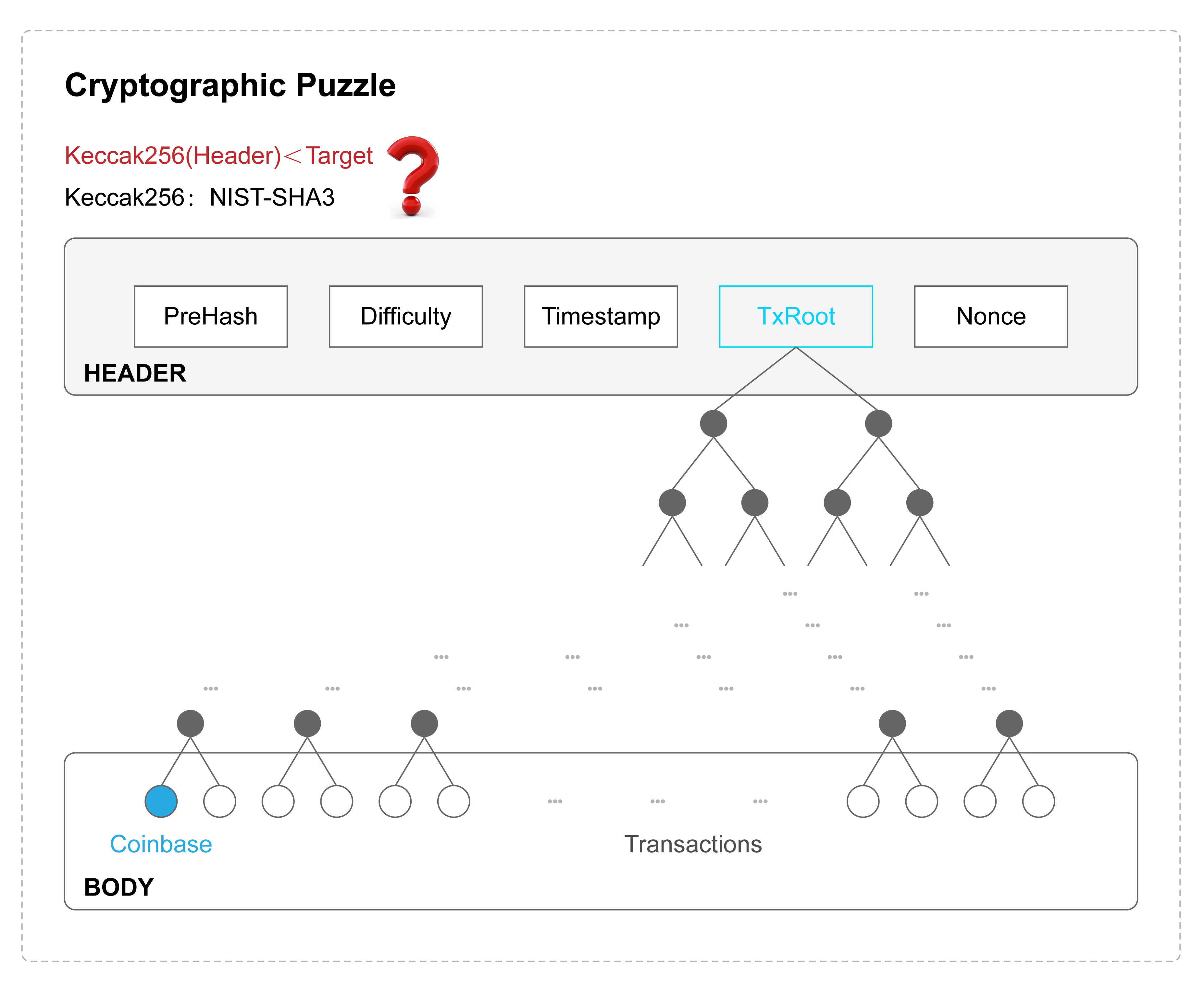}
  \caption{The cryptographic puzzle of Proof of Work}\label{puzzle}
\end{figure}

Bitcoin uses the Hashcash proof of work system \cite{HC1}. The original Bitcoin proof of work algorithm is based on the SHA256 hash function which uses the Merkle-Damgard construction, however, this function has the length-extension weakness \cite{SHA256}. We suggest using the SHA3-Keccak256 hash function instead \cite{SHA3}, as this hash function uses the sponge and duplex construction, which is not sensitive to length-extension attacks, thus making protocols that utilize it more robust.

We outline the procedure used by Bitcoin to generate a block with valid proof of work in detail below:
\begin{enumerate}
\item  Each miner collects the transactions into a memory pool separately. When mining begins, miners pick transactions from the pool and use them to generate a Merkle tree\cite{MT}.
\item  The miner generates a temporary block which has everything determined except the nonce. This temporary block consists of a header and body. The header contains the previous block header hash, difficulty, transaction Merkle root (obtained in Step 1), and timestamp, while the body contains the transaction list.
\item  The miner will continuously put different nonces into the temporary block and compute the hash of the block header until either he finds a valid proof of work or the miner receives a new block from the network, meaning another miner has already validated the current block. If the miner finds a valid proof of work, he will broadcast the new block. If the entire searching space of the nonce has been tried, the miner will return to Step 1, modify the template by changing the transaction or timestamp, and continue trying.
\end{enumerate}

\begin{figure}[htbp]
\centering
\includegraphics[width=5in]{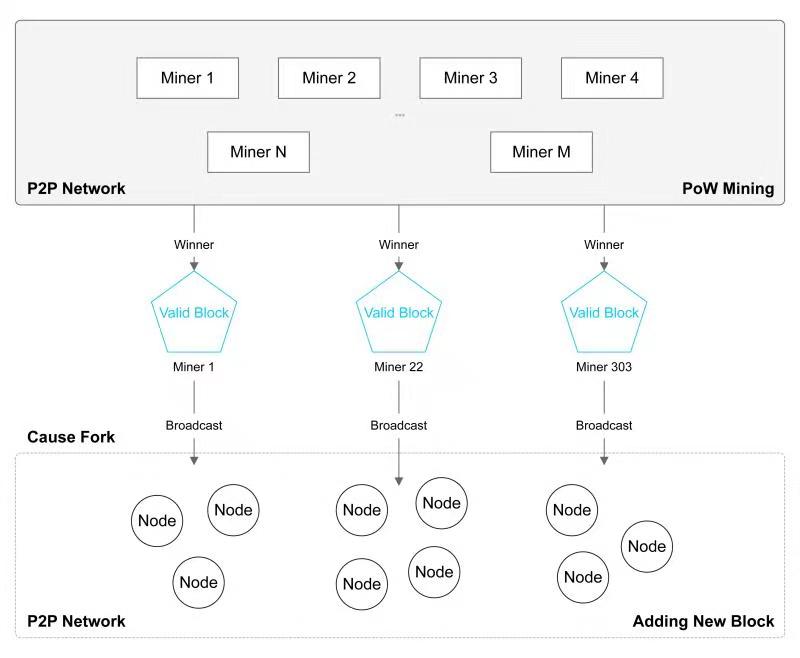}
\caption{An overview of the original Proof of Work protocol. Miners work individually and forks can occur if multiple miners solve the cryptographic puzzle. }\label{pow}
\end{figure}

  Fig \ref{pow} depicts a traditional Proof of Work procedure. We can see that each miner is working separately from other miners. If they are working in the same transaction list, there is a high probability of repeating work, resulting in a waste of hash computing power. Moreover, multiple miners create blocks with the same preceding block, resulting in several valid blocks at the same height in the chain. PoW allows the possibility of forks, where several different blockchains have the same length, yet no chain supersedes other chains. Miners can choose blocks from any branch as their preceding block. The Bitcoin protocol defines a regulation on which chain miners should mine in order to resolve forks. The criterion is that only the longest chain will be committed by the P2P network, as the longest chain signifies it consumed the most mining power to generate. However, the Bitcoin consensus protocol is not incentive-compatible; Eyal and Sirer \cite{ES1} describe an attack, called selfish mining, in which colluding miners can obtain extra revenue. This attack increases the size of the colluding group until it becomes the majority. They propose one modification for the protocol to mitigate this attack: when receiving competing branches of the same length, the miner should broadcast all of the branches and pick one to mine randomly. Decker and Wattenhofer \cite{DW1} estimate that accidental bifurcation occurs on average about once every 56 blocks. The faster blocks are produced, the higher the chance that forks will form, which negatively impacts the consistency of the chain. 
  
  In section 3 we propose a map-reduce mining mechanism to address the problem of repeated work in PoW. Any miner can submit a valid block by broadcasting it to the peer to peer (P2P) network, then other miners verify its validity, and if it is valid, they will commit it. In section 4, we introduced a PBFT verification step after map-reduce mining, which aims to solve the consistency problem of the original PoW by disallowing forking through providing deterministic finality, instead of the probabilistic finality used by PoW. An overview of the DPoW protocol is illustrated in Fig \ref{dpow}.
  
\begin{figure}[htbp]
  \centering
  \includegraphics[width=5in]{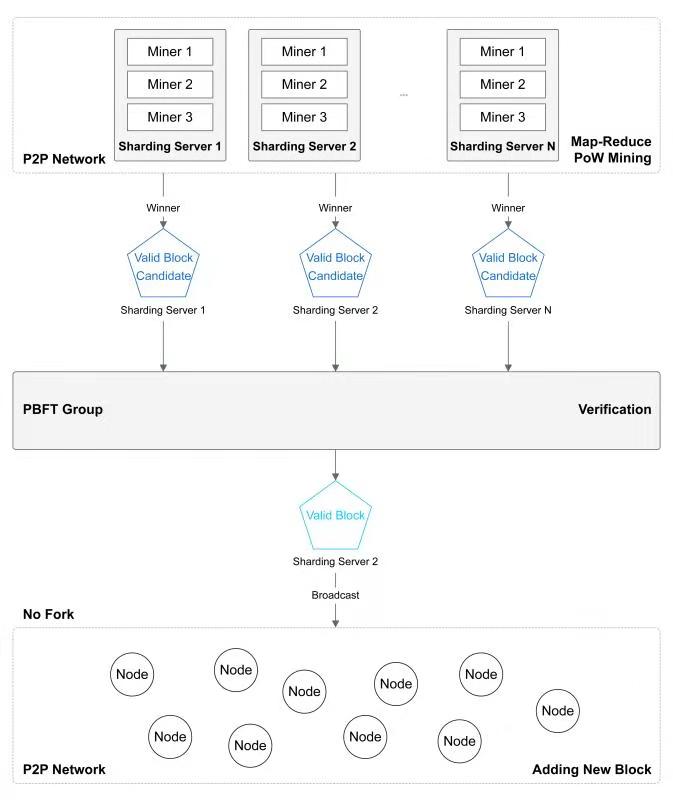}
  \caption{An overview of the Deterministic Proof of Work protocol. Miners are grouped by sharding servers. Sharding servers solve the cryptographic puzzle and broadcast the block to a verifier group. Only one block can be verified and broadcasted to the P2P network at a time.}\label{dpow}
\end{figure}

\section{Map-Reduce Proof of Work}

The original proof of work mining mechanism requires all miner to brute-force the cryptographic puzzle individually. This high amount of work repetition results in a major loss in efficiency. In order to improve the efficiency of this process, we design a map-reduce proof of work procedure. Let us first introduce a new type of node called a ``sharding server.'' In the original proof of work mining system, miners have two major tasks. First, they must construct a temporary block as the input for mining. Second, they attempt to fill the nonce in the temporary block to check if it is valid. In our procedure, we let the sharding server build the temporary block from the memory pool and miners will only need to find the nonce for the specified temporary block. The necessary steps are described in detail below:

\begin{enumerate}
\item  The sharding server will map the entire searching space into several shards using a non-overlapping method, then distribute the individual shards to the miner nodes working for it.

\item  The miner nodes will work in their searching shards and attempt to obtain the proof of work. The miner should also frequently listen to the sharding server within a reasonable interval of time for new updated shards. If the miner obtains the proof of work, he will submit it to the sharding server.

\item  The sharding server will then submit the proof of work and wait for verification.
\end{enumerate}

In Algorithm \ref{MRPOW}, we present the pseudocode of the Map-reduce PoW algorithm.

\begin{algorithm}
  \caption{Map Reduce PoW}
  \label{MRPOW}
  \begin{algorithmic}[1]
   \State master.start()
   \State worker.start()
   \State master.register(worker) \Comment{notify master about the worker}
   \State master.dispatch(work)  \Comment{distribute shards to worker}
   \State notFound = true    \Comment{miner starts mining}
   \While{notFound}
     \State worker.change($coinbase$) \Comment{worker changes coinbase data to make a trial}
     \State $pass \gets worker$.seal()    \Comment{worker computes hash and returns true on success}
     \If{$pass = true$}
        \State $notFound \gets false$
     \EndIf
   \EndWhile
 \State worker.submit($work$) \Comment{worker submits proof of work to master}
 \State $pass \gets$ master.check($work$) \Comment{master verifies proof of work}
 \If{$pass = true$}
  \State  master.broadcast($block$) \Comment{master broadcasts the valid block}
 \EndIf
  \end{algorithmic}
\end{algorithm}
The challenge lies in how to design the mapping algorithm to avoid overlaps. The nonce is normally 4 bytes (or 32 bits), which means the searching space contains $2^{32}$ possible nonces. Most miners can cover this space in seconds, thus we cannot simply split the entire range of nonces. We designed a method to map the searching space into separate non-overlapping spaces through the coinbase transaction. The sharding server will listen to the P2P network and package the transaction, however, it will not generate the entire transactions' Merkle tree. Instead, it will create the Merkle branch of the coinbase transaction. Then the sharding server will distribute the previous block header hash, difficulty, timestamp, Merkle branch of the coinbase transaction, and the receiver address of the coinbase transaction to miners. Nonces are not distributed. The sharding server will then wait for responses from the miner servers. It will not wait endlessly, and will instead map new shards and distribute them in the condition that no miner submits a proof of work in a certain period, or another sharding server obtains a verified proof of work first.

The miner gets a distributed shard from the sharding server and then works on this shard, then the miner will produce the coinbase transaction by concatenating the extra data and the address of the coinbase transaction's receiver, provided by the sharding server. This extra data will use the miner's unique ID (for example, their bitcoin address) as a prefix, followed by any extra byte array. After the miner produces the coinbase transaction, he can use it to generate the transaction root hash using the Merkle branch given by the sharding server. The miner will then attempt to adjust the nonce to get a valid proof of work. If the miner is unsuccessful in their attempt, he can adjust the extra data and generate a new coinbase transaction until he is successful. Once the coinbase is generated, the entire searching space is modified, and the miner can then try the nonce in this new searching space. Since each miner has a different unique ID, and since the extra data begins with their ID, the extra data will not overlap with other miners, thus we can expect that the Merkle root search by different miners will have a very minimal probability of collision.

The collision probability can be analyzed through the birthday model. Suppose there are $m$ people and $N$ days per year. We can use $P(m,N)$ to denote the probability that there are at least two people with same birthday, as described in Equation \ref{EQ1}.
\begin{equation}
\label{EQ1}
\begin{aligned}
P(m,N)&  = 1 - \frac{{N!}}{{(N-m)! \cdot {N^m}}}= 1 - \prod\limits_{i = 1}^{m - 1} ({1 - \frac{i}{N}}) \\
& \ge 1 - \prod\limits_{i = 1}^{m - 1} {{e^{ - \frac{i}{N}}}} {\rm{ = 1}} - {{\rm{e}}^{ - \frac{{m(m - 1)}}{{2N}}}}.
\end{aligned}
\end{equation}
In our model, we can set $p(m)$ as the probability that a collision happens given $m$ different coinbase messages. Since the hash function's output length is $256$ bits, there are $2^{256}$ different possible outputs. We can treat the collision event as if there are at least two people (miners) having the same birthday (a collision), where $N$ is $2^{256}$. This probability estimation is represented in Equation \ref{EQ2}.

\begin{equation}
\label{EQ2}
\begin{aligned}
p(m) \ge {\rm{1}} - \exp ( - \frac{1}{2} \cdot \frac{{{m^2}}}{{{2^{256}}}}) \approx \frac{1}{2} \cdot {(\frac{m}{{{2^{128}}}})^2}.
\end{aligned}
\end{equation}
The values of $m$ and their lower bounds of collision probability are displayed in Table \ref{tab1}.

\begin{table}[htbp]
\centering
\caption{The relationship between collision probability and number of trials.}
\label{table}
\begin{tabular}{|c|c|c|c|c|c|c|}
\hline
$m$   & $2^{16}$   &$2^{32}$     &$2^{64}$    &$2^{96}$    &$2^{127}$     &$2^{128}$ \\
\hline
$p(m)$ & $2^{-225}$ &$2^{-193}$  &$2^{-129}$  &$2^{-65}$    &$0.125$      &$0.5$ \\
\hline
\end{tabular}
\label{tab1}
\end{table}

From Table \ref{tab1}, we can see that it will take $2^{128}$ trials before the chance that some pair of miners have a hash collision becomes at least $50\%$, thus the chance of having a collision is negligible. Thus, if each miner uses a unique ID as their coinbase extra data prefix, their Merkle root will have little chance of overlap, and thus neither will their searching space.

The blockchain will only reward the sharding server node through coinbase transaction. In order to incentivize miners, the sharding server must distribute the reward using an appropriate method. This can be done either in chain or off chain, however, the details of this system will not be discussed in this paper as it is not a major concern of the main chain consensus protocol.

Once the sharding server obtains a valid proof of work, it will broadcast the valid block to a randomly selected verifier group and wait for verification. If the block gets verified, all nodes in the P2P network can then replicate their local state machine based on this block. The verification process will be detailed in the next section.
  
\section{Byzantine Fault Tolerance Verification}
In the last section, we proposed the Map-reduce PoW mining algorithm to generate a valid candidate block. After the sharding server generates the candidate block, the server must submit it to a verifier group and await verification. When the candidate block is received by the verifier group, they will run a Practical Byzantine Fault Tolerance (PBFT) protocol to ensure the consistency and correctness of the verification result in asynchronous environments \cite{PBFT2}. The verifier group is elected through a verifiable random function (VRF) \cite{VRF} and the result is recorded in the main chain periodically. The election procedure will not be covered in this paper. The verifier group will verify a certain amount of blocks before it is replaced by a new elected verifier group. The timeframe within which each chosen verifier group performs all of its tasks is called one period, and the timeframe within which a single block verification procedure is conducted is called a round. For every round, each verifier in the group takes turns becoming the chosen primary verifier. This primary verifier initializes the verification protocol and also picks one block among the valid block candidates received from different sharding servers. Then, the primary verifier proposes this block to the rest of the verifier group, who will then execute the PBFT protocol on the block.

The original PBFT protocol \cite{PBFT1} requires some modification in order to be adapted for our system. The verification protocol procedure contains four steps, shown in Figure 4. In a normal case, the protocol begins with the "propose" stage. Through the "prevote" and "precommit" stages, verfiers vote and reach consensus on either a proposed block or an empty block. After this stage, they undergo the "commit" stage where the verified block is broadcast to the rest of the network. The blockchain height is thus increased with the addition of the new block and consensus continues into a new cycle of block selection and verification. If any errors occur, such as in the case that the primary verifier proposes an invalid block, consensus will follow the red path shown in Figure 4, skip the commit stage (thus leaving the blockchain height unchanged), and enter a new round of consensus.

\begin{figure}[htbp]
  \centering
  \includegraphics[width=5in]{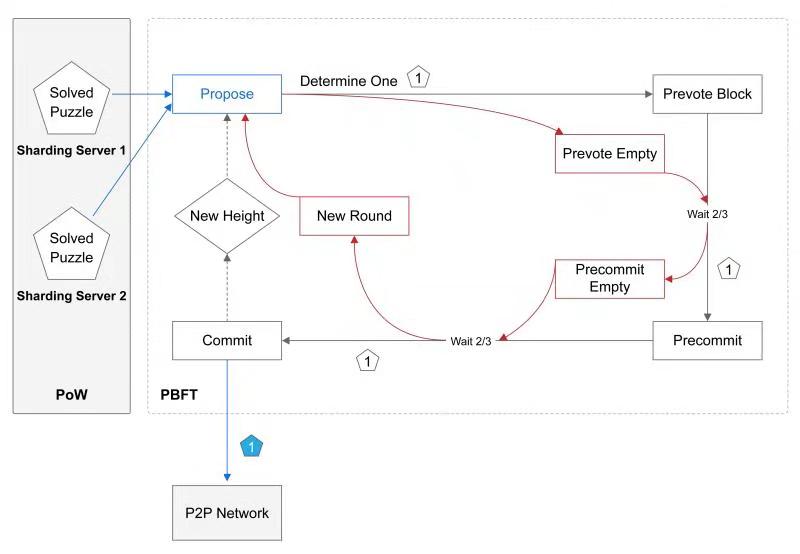}
  \caption{DPoW adapted Practical Byzantine Fault Tolerance.}\label{pbft}
\end{figure}
The top-level procedure that implements PBFT is shown in Algorithm 2. This procedure takes in a context, $ctx$, which captures the current state of the blockchain. The procedure will initialize the height and round $<h,r>$ as $<chain height + 1, 0>$, and will set the lockedBlock to "empty," then begin the consensus loop. The stages of the verification process are detailed as follows.

\begin{enumerate}
\item "Propose." The Propose() function is first called by verifiers. The verifiers first wait for a block proposal by the primary verifier. If the lockedBlock is not empty, the proposer proposes this locked block. Otherwise, the proposer collects a new mined blocks from sharding servers within a reasonable time period. The first valid block is then chosen and a block proposal is generated containing the height, round, proposed block, and signature. The primary verifier then generates the signature, fills the proposal, and broadcasts it to other verifiers. Upon receiving the block proposal, verifiers must first check if the signature is valid by checking if it was indeed produced by the primary verifier. Furthermore they need to verify if the block itself is valid. If is invalid, or no proposal was received upon time out, they will instead default on an empty block proposal.

\item "Prevote." Each verifier obtains a valid or empty proposal from the first stage. A vote for that proposed block hash or empty block hash is then generated and broadcasted to the group. While this occurs, verifiers also collect the votes received from other verifiers. This procedure is responsible for ensuring vote validity by checking height, round, and signature. Finally, a vote count is performed, and the block hash which obtains more than 2/3 of votes is chosen and returned. If there is no consensus on a hash, meaning no hash surpassed the threshold amount of votes, an empty block hash will be returned instead.

\item "Precommit." In this stage, verifiers broadcast their precommit vote to the hash returned by the "Prevote" stage. A vote count is performed once again, and if more than 2/3 of precommit votes are received for a block hash, the hash is returned, otherwise the empty block hash is returned.

\item "Commit." If the block hash returned by the "Precommit" stage is not empty, it must be equal to H(lockedBlock). The lockedBlock is then committed and the height, round, and lockedBlock is updated. If the blockHash instead represents an empty block, then simply increment the round by 1 and continue. 

\end{enumerate}

It is crucial to ensure that consensus is reached on only one valid block for each blockchain height. It is possible to have two different blocks become verified at any given height, in the case that some verifiers support both block A and block B in two different rounds. Take for example if there are four verifiers: A, B, C, and D. A has the votes for Block X from A, B, and C and thus commits Block X. If A's vote is not received by the rest of the group in time due to a network delay, then B, C, and D only receive B and C's votes on Block X and enter a new round. B, C, and D then reach consensus on some block Y instead. As a result, both Block X and Y appear as valid to the network. To avoid this problem, when a user observes more than a 2/3 majority vote on a block in the Prevote stage, the user will lock this block and precommit it. If this block is not commited in this round, such as in the case of verifier B and C, they will begin a new round. If B or C becomes the primary verifier, they will then re-propose Block X in order to ensure the group reaches consensus on the same block as A. If instead verifier D becomes the primary verifier, a new Block Y is proposed, however, B and C will stick to Block X and instead vote on an empty block hash to be safe.

\begin{algorithm}
\caption{Practical Byzantine Fault Tolerant}
\label{BFT}
\begin{algorithmic}[1]
\Procedure{consensus}{$ ctx $}
\State $h\gets ctx.height + 1$   \Comment{start from the next block in the chain}
\State $r\gets 0$
\State $empty \gets nil$
\State $lockedBlock \gets empty$
\While {$r < MaxStep$} \Comment{if r=MaxStep, assume network problem}
  \State $proposeBlock \gets $ Propose($ctx,h,r,lockedBlock$)
  \State $lockedBlock \gets $  Prevote($ctx,h,r,$H($proposeBlock$))
  \State $blockHash \gets $  Precommit($ctx,h,r,lockedBlock$)
  \If {$blockHash = H(lockedBlock)$} 
  \State Commit($ctx,lockedBlock$)   \Comment{consensus reached}
  \State $h++$
  \State $r\gets 0$
  \State $lockedBlock\gets empty$
  \Else	\Comment{no consensus reached in this round}
  \State $r++$ \Comment{start a new round}
  \EndIf
\EndWhile
\EndProcedure
\end{algorithmic}
\end{algorithm}

\section{Security Discussion}

In this section we will analyze the security of our protocol. Suppose there are in total $N$ parties in our network. Among these parties, there are $T$ parties controlled by the attacker. In order to simplify the security model, we assume all parties have the same hash computing power. If one party has twice the hash computing power of other parties, we can simply treat it as two parties. We evaluate the possibility that the attacker can successfully deploy a double spending attack.

\begin{lemma} \label{lemma1}
The PBFT verification step can commit a proposal if more than $2/3$ of all verifiers are in agreement.
\end{lemma}

\begin{proof}
See \cite{PBFT1} for details.
\end{proof}

\begin{thm}\label{thm1}
Assuming an honest majority, the Deterministic Proof of Work scheme is secure against double spending attacks.
\end{thm}

\begin{proof}
We first observe how an attacker can deploy a double spending attack \cite{DBS}. First, the attacker must corrupt and control more than $2/3$ of the verifier group from Lemma \ref{lemma1}. This lies in the fact that honest verifiers will never sign a block with the same height twice, thus the attacker must ensure more than 2/3 of verifiers will sign whatever block the attacker chooses. Then in the second step, the attacker must gain enough hash computing power to suppress an old verified block. If both of these tasks can be achieved, then the attacker can send their transaction and get what he requires from the receiver. Once this has occurred, the attacker can then attempt to generate another block with the same height as the block that must be suppressed, and then ensure that it can be verified by the verifier group. If the new block suppresses the old block and is successfully committed to the P2P network, a double spending attack will have been successfully conducted. We can now analyze this attack's probability of success in our protocol. We denote this probability as $P_{att}$.

First, we calculate the probability that more than 2/3 of verifiers are from an attacking party. Suppose the total number of verifiers is $M$. Let $X_i$ be the random variable equal to 1 if the $i$-th verifier is from an attacking party and 0 otherwise. Since the verifier is randomly chosen from the total parties, it is easy to see that $X_i$ is a Bernoulli random variable. We have that $P({X_i} = 1) = \frac{T}{N}$ and $P({X_i} = 0) = \frac{N-T}{N}$. Let $X = \sum\limits_{i = 1}^M {{X_i}}$. We can see that X follows a binomial distribution with parameters $(\frac{T}{N},M)$. Using multiplicative Chernoff bound, we have $P(X \ge (1{\rm{ + }}\delta )\frac{{MT}}{N}) \le {e^{-\frac{{{\delta ^2}\frac{{MT}}{N}}}{3}}}$ for any $0 \le \delta  \le 1$ and $P(X \ge (1{\rm{ + }}\delta )\frac{{MT}}{N}) \le {e^{-\frac{{{\delta}\frac{{MT}}{N}}}{3}}}$ for any $\delta  \ge 1$. We denote the probability $P(X \ge [\frac{2M}{3}])$ as $P_f$. This is the probability that the PBFT verification step become insecure. Set $\delta =\frac{2N}{3T}-1$, we have that ${P_f} \le {e^{ - \frac{{{{\left( {\frac{{2N}}{{3T}} - 1} \right)}^2}\frac{{MT}}{N}}}{3}}}$ for $\frac{1}{3}N \le T \le \frac{2}{3}N$ and ${P_f} \le {e^{ - \frac{{\left( {\frac{{2N}}{{3T}} - 1} \right)\frac{{MT}}{N}}}{3}}}$ for $T < \frac{1}{3}N$.

Now we assume that the attacker already controls more than 1/3 of verifiers in the verifier group. What the attacker must do next is to suppress the verified block. We denote the probability of the attacker's success in this step as $P_s$. The race between the honest majority and attacker can be characterized as a Binomial Random Walk. From \cite{BTC1}, we know that the probability of an attacker catching up is analogous to a Gambler's Ruin problem. Thus, we have:
\begin{enumerate}
\item Probability an honest node finds the block is $\frac{N-T}{N}$
\item Probability the attacker finds the block is $\frac{T}{N}$
\item Probability the attacker can catch up from $z$ blocks behind is\\$\left\{ {\begin{array}{*{20}{c}}
{1,T \ge \frac{N}{2}}\\
{{{(\frac{T}{{N - T}})}^z},T < \frac{N}{2}}
\end{array}} \right\}$
\end{enumerate}

Now suppose the attacker begins mining a secret block immediately after his transaction is committed to a block B, and the attacker get what he wants $z$ blocks after block B. Then, the length of the secret chain of the attacker will be a Poisson distribution with expected value $\lambda \frac{T}{{N - T}}$. For simplicity, we assume the period of the verification is longer than $z$, meaning that we do not need to account for a change to the verifier group. Otherwise, the attacker must corrupt the verifier group in a successive period, which increases the difficulty of the attack.

Now suppose the length of the secret chain is $k$ blocks. We will deal with it in two cases. In the first case, given that $\frac{T}{N} < \frac{1}{2}$, if $k \le z$, the number of blocks that the attacker must catch up with is $z-k$, thus the probability that he can catch up is given by ${{(\frac{T}{{N - T}})}^z}$. Otherwise, the probability that the attacker can catch up is $1$ since his secret chain is already longer. Accumulating the probability for each k, we obtain:

\begin{equation}
\label{EQTT}
\begin{aligned}
p_s= & \sum\limits_{k = 0}^z { \frac{{{\lambda ^k}{e^{{\rm{ - }}\lambda }}}}{{k!}} \cdot {{\left( {\frac{T}{{N - T}}} \right)}^{(z - k)}}}  + \sum\limits_{k = z}^\infty  {\frac{{{\lambda ^k}{e^{{\rm{ - }}\lambda }}}}{{k!}}} \\
{\rm{ = }} &1{\rm{ - }}\sum\limits_{k = 0}^z {\frac{{{\lambda ^k}{e^{{\rm{ - }}\lambda }}}}{{k!}} \cdot \left( {1{\rm{ - }}{{\left( {\frac{T}{{N - T}}} \right)}^{(z - k)}}} \right)},
\end{aligned}
\end{equation}
In the second case, when $\frac{1}{2} \le \frac{T}{N} < \frac{2}{3}$, the attacker has more than $50\%$ of hash computing power. Similarly we find $P_s=1$.

Now we can obtain the probability $P_{att}$ as follows:

\begin{equation}
\label{EQTH}
\begin{aligned}
{{\rm{P}}_{att}}
& {\rm{ =  }}{{\rm{P}}_f} \cdot {{\rm{P}}_s}\\
& \le \left\{ {\begin{array}{*{20}{c}}
{{e^{ - \frac{{{{\left( {\frac{{2N}}{{3T}} - 1} \right)}^2}\frac{{MT}}{N}}}{3}}},\frac{1}{2} \le \frac{T}{N} \le \frac{2}{3}}\\
{\left( {1{\rm{ - }}\sum\limits_{k = 0}^z {\frac{{{\lambda ^k}{e^{{\rm{ - }}\lambda }}}}{{k!}} \cdot \left( {1{\rm{ - }}{{\left( {\frac{T}{{N - T}}} \right)}^{(z - k)}}} \right)} } \right) \cdot {e^{ - \frac{{{{\left( {\frac{{2N}}{{3T}} - 1} \right)}^2}\frac{{MT}}{N}}}{3}}},\frac{1}{3} \le \frac{T}{N} < \frac{1}{2}}\\
{\left( {1{\rm{ - }}\sum\limits_{k = 0}^z {\frac{{{\lambda ^k}{e^{{\rm{ - }}\lambda }}}}{{k!}} \cdot \left( {1{\rm{ - }}{{\left( {\frac{T}{{N - T}}} \right)}^{(z - k)}}} \right)} } \right) \cdot {e^{ - \frac{{\left( {\frac{{2N}}{{3T}} - 1} \right)\frac{{MT}}{N}}}{3}}},\frac{T}{N} \le \frac{1}{3}}
\end{array}} \right.
\end{aligned}
\end{equation}

From Inequality \ref{EQTH}, we can observe that even if the attacker controls more than $50\%$ of hash computing power, the success of a double spending attack is still not guaranteed. If the attacker wishes to successfully deploy the double spending attack with 100\% probability of success, he must have at least 2/3 of total hash power. Furthermore, in the case that the attacker controls less than $50\%$ of hash computing power, the success probability of such an attack in our protocol is still lower than that of the original PoW algorithm given in \cite{BTC1}, which proves that our protocol is more secure.

\end{proof}

\section{Experiment}
In this section, we conducted experiments to demonstrate the performance and security of our consensus protocol. We implement our protocol using the Go language.

First, we tested the efficiency of Map-reduce PoW by creating two groups to mine for the proof of work given the same hash puzzle by having them work on the same temporary block. In Group 1 (without sharding), we used 7 sharding servers with each server relying on a single miner. Each miner solved the puzzle for the entire searching space, which is similar to the original PoW scheme. In Group 2 (with sharding), we used 1 sharding server with 7 miners working for it simultaneously. For each group, we recorded the time taken to generate a valid proof of work using the same hash difficulty. We randomized the level of difficulty for each trial and ran the experiments on 1000 hash puzzles.

In Fig. \ref{1a}, the times taken to generate the proof of work for all $1000$ hash puzzles are shown in the scatter plot. As expected, we can see that the group with sharding was much faster in comparison to the group without sharding, with only a single trial exceeding $30$ seconds.

In Fig. \ref{2a}, we present a box plot of consensus times in both groups. The distribution of time from the group with sharding has a much smaller spread compared to the group without sharding. The interquartile range (IQR) of consensus times for the group with sharding ($1$ to $8$ seconds) was $3.7$ times smaller than the group without sharding ($5$ to $31$ seconds). We observed an average time of $5.37$ seconds per puzzle with sharding compared to an average of $21.62$ seconds without sharding. We then removed the $191$ observations with consensus time $0$ and log-transformed the remaining times to meet the normality assumption. The Welch two sample t-test infers a significantly different mean consensus time between the two groups ($t = -21.8$, $df = 1514.3$, $\textrm{p-value} < 0.01$**).
The results suggest that our Map-reduce PoW protocol ensures that block hash puzzles are successfully divided among miners such that there is very minuscule overlap of work, thus allowing a much larger searching space per second with decreased computation power compared to the original PoW protocol in which many overlaps result in slow and inefficient computation. Additionally in our protocol, while holding difficulty constant, as the number of miners solving a hash puzzle increases, the more distributed the work becomes, and thus the time to calculate each puzzle decreases.
\begin{figure}[htbp]
\centering\includegraphics[width=3.5in]{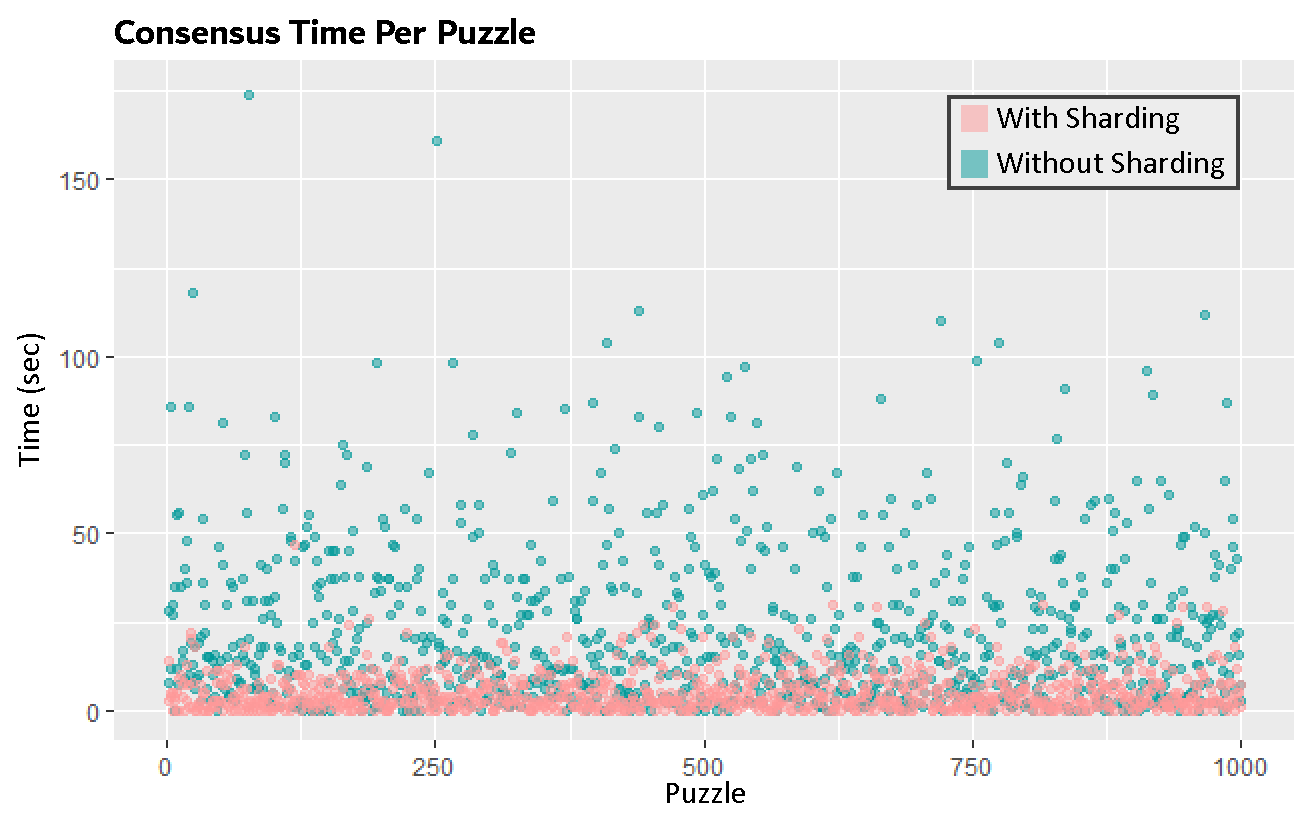}
\caption{Scatter plot comparing times taken to generate the work per puzzle with and without sharding.}\label{1a}
\end{figure}

\begin{figure}[htbp]
\centering\includegraphics[width=3.5in]{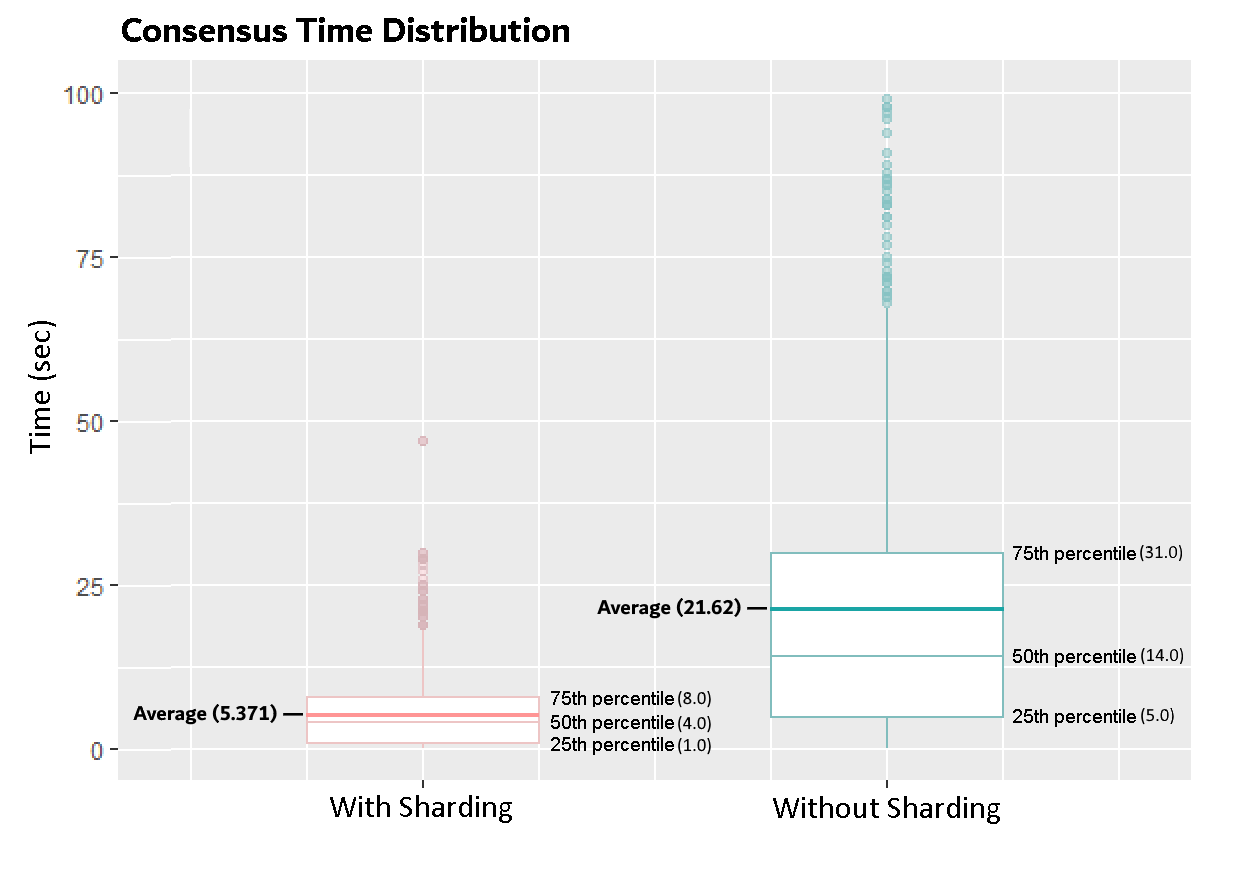}
\caption{Box plot comparing distributions of consensus times with and without sharding.}\label{2a}
\end{figure}

Next, we conducted an experiment to test the PBFT verification protocol. We set up the experiment by creating six groups. A valid block was assigned to Group A1, A2, and A3, and an invalid block was assigned to Group B1, B2, and B3. Each group consisted of $4$ verifiers that were either good (meaning they can always validate a block correctly) or bad (meaning they may validate a block incorrectly). Group A1 and B1 each contained $1$ bad verifier and $3$ good verifiers. Group A2 and B2 each contained $2$ bad verifiers and $2$ good verifiers. Group A3 and B3 each contained $3$ bad verifiers and $1$ good verifier. This setup establishes that Group A1 and B1 have less than $1/3$ bad verifiers, Group A2 and B2 have between $1/3$ and $2/3$ bad verifiers, and Group A3 and B3 have more than $2/3$ bad verifiers.

In the experiment, we assumed bad verifiers could conspire together to alter the validation process. We then ran the experiment and recorded whether the groups marked blocks as valid (V), or invalid (I). This procedure was then run 10 times for each group. The results of this experiment were recorded in Table \ref{tab2}.

\begin{table}[htbp]
\centering
\caption{Block validation results for valid (V) and invalid (I) block trials given different amounts of bad verifiers.}
\label{table}
\begin{tabular}{c|cccccccccc}
  Trial    & 1 & 2 & 3 & 4 & 5 & 6 & 7 & 8 & 9 & 10 \\ \hline
  Group A1 & V & V & V & V & V & V & V & V & V & V \\
  Group A2 & V & I & V & V & I & I & V & V & V & I \\
  Group A3 & I & I & I & I & I & I & I & I & I & I \\
  Group B1 & I & I & I & I & I & I & I & I & I & I \\
  Group B2 & I & I & I & I & I & I & I & I & I & I \\
  Group B3 & V & V & V & V & V & V & V & V & V & V \\
\end{tabular}
\label{tab2}
\end{table}

We can see from the table that in Group A1 and B1, every block was successfully identified as valid and invalid respectively. In Group A2, given only valid blocks, the results were mixed, indicating that the group may not achieve consensus. In Group B2, given only invalid blocks, it successfully identified every block as invalid. In Group A3 and B3, we see that since the bad verifiers could conspire together and can control the verification process, they were able to force a false validation to be accepted every time.

We can confirm from these results that, given a group of verifiers, if there are less than $1/3$ bad verifiers, all blocks that are processed by the verifier group will be verified correctly. If there are between $1/3$ and $2/3$ bad verifiers, invalid blocks submitted to the group will always be correctly identified, however, the group may not be able to achieve consensus on valid blocks. Finally, if there are more than $2/3$ of verifiers conspiring in the PBFT protocol, invalid blocks submitted to the group could be falsely marked valid, and vice versa for valid blocks. From this, we can infer that if a group consists of more than $2/3$ bad verifiers, these verifiers can control the verification process.

\section{Summary and Future Work}

In this paper, we have described a Deterministic Proof-of-Work consensus algorithm, which is a hybrid approach that presents a new Map-reduce PoW mining algorithm and combines it with PBFT verification. We have analyzed the security model of this new consensus protocol and estimated the probability of successful double spending attacks. Our protocol can successfully resist 51\% attacks and also increase this threshold to 66.6\%. We also demonstrated that our protocol can achieve greater security and consistency while keeping high level decentralization through our two experiments.

It is quickly becoming evident that no pure consensus protocols can optimize the DCS triangle efficiently. For example, DPoS is criticized for its decreased decentralization, the consistency of PoS is affected by the Nothing-at-Stake problem, and the PoW protocol suffers from its major scalability problem. Additionally, PBFT alone has network scalability problems that result in it being used primarily for consortium chains. Compared to pure consensus protocols, hybrid consensus protocols, such as ours, prove to be far more capable of optimizing the DCS triangle while also providing forward security and are more practical and efficient as a result.

There is still much work to do on improving our system.

Firstly, we can incorporate a reputation score in the verifier election process, which can in turn raise the difficulty of controlling verifiers in the verifier group, thus increasing security against attackers. This score can be combined with the random seed to scale the probability for a candidate to be elected as a verifier. More work needs to be conducted in regards to this system in the future.

The DPoW protocol can also utilize IBLT (Invertible Bloom Lookup Tables) and a bloom filter \cite{IBLT} when broadcasting blocks within the P2P network. This can improve the performance of the protocol by reducing the network cost significantly when relaying blocks between nodes.

\end{document}